\documentclass[conference,letterpaper]{IEEEtran}

\usepackage[utf8]{inputenc} 
\usepackage[T1]{fontenc}
\usepackage{url}
\usepackage{ifthen}
\usepackage{cite}
\usepackage[cmex10]{amsmath} 
\usepackage{hyperref}

\usepackage{amsfonts, amssymb, amsthm}
\usepackage{array}
\usepackage{tikz, color, soul}
\usetikzlibrary{plotmarks}
\usepackage{mathtools}

\usepackage{mathtools}

\newtheorem{theorem}{Theorem}

\newtheorem{lemma}{Lemma}

\newcommand{\y}{\ensuremath{\mathbf{y}}}
\newcommand{\z}{\ensuremath{\mathbf{z}}}
\newtheorem{definition}{Definition}

\newtheorem{remark}{Remark}

\begin{document}
	\title{Optimal Multistage Group Testing Algorithm for 3 Defectives} 
	
	\author{%
		\IEEEauthorblockN{Ilya~Vorobyev}
		\IEEEauthorblockA{
			Center for Computational and Data-Intensive Science and Engineering, \\
			Skolkovo Institute of Science and Technology\\
			Moscow, Russia 127051}
		\IEEEauthorblockA{
			Advanced Combinatorics and Complex Networks Lab, \\
			Moscow Institute of Physics and Technology\\  Dolgoprudny, Russia 141701}
		\IEEEauthorblockA{\textbf{Email}: vorobyev.i.v@yandex.ru}
	}

	\maketitle
	
	\begin{abstract}
		Group testing is a well-known search problem that consists in detecting of $s$
		defective members of a set of $t$ samples by carrying out tests on properly chosen subsets
		of samples. In classical group testing the goal is to find all defective elements by
		using the minimal possible number of tests in the worst case. 
		In this work, a multistage group testing problem is considered. Our goal is to construct a multistage search procedure, having asymptotically the same number of tests as the optimal adaptive algorithm. We propose a new approach to designing multistage algorithms, which allows us to construct a 5-stage algorithm for finding 3 defectives with the optimal number $3\log_2t(1+o(1))$ of tests.
	\end{abstract}

	
	\section{Introduction}\label{sec::intro}
	
	The group testing problem was introduced by Dorfman in~\cite{dorfman1943detection}. Suppose that we have a large set of samples, some of which are defective. Our task is to find all such elements by performing special tests. Each test is carried out on a properly chosen subset of samples. The result of the test is positive if there is at least one defective element in the tested subset; otherwise, the result is negative.  In this work we consider the noiseless case, i.e., the outcomes are always correct. We aim to design an algorithm that finds all defective elements using as few tests as possible.

	Two types of algorithms are usually considered in group testing. {\it Adaptive} algorithms can use the results of previous tests to determine which subset of samples to test at the next step. In {\it non-adaptive} algorithms all tests are predetermined and can be carried out in parallel.
	
	In this paper, we study multistage algorithms, which can be seen as a compromise solution to the group testing problem. An algorithm is divided into $p$ stages. Tests from the $i$th stage may depend on the outcomes of the tests from the previous stages. 
	
	We consider a problem, in which the total number of defective elements is equal to 
	$s$. Let $N_p(t,s)$ 
	be the minimal worst-case total number of tests needed to find all $s$ 
	defective members of a set of $t$ samples using at most $p$ stages; $N_{ad}(t, s)$ 
	stands for the minimal number of tests for adaptive algorithms. 
	
	
	In many applications, it is much cheaper and faster to perform tests in parallel. Unfortunately, non-adaptive algorithms require much more tests than adaptive ones. It  is known~\cite{d1982bounds,ruszinko1994upper,furedi1996onr} that for fixed $s$ non-adaptive algorithm needs at least $N_1=\Omega\left(\frac{s^2\log_2t}{\log_2s}\right)$ tests, whereas with the adaptive algorithm it is sufficient to use only $N_{ad}(t, s)=s\log_2t(1+o(1))$, $t\to\infty$, tests. 
	Rather surprisingly, for 2-stage algorithms it was proved that ${O(s \log_2 t)}$ tests are already sufficient~\cite{de2005optimal,rashad1990random,d2014lectures}. This fact emphasizes the importance of multistage algorithms.
	
	In this paper, we are interested in the constant
	$$
	C_p(s)=\varlimsup_{t\to \infty}\frac{N_p(t, s)}{\log_2t}
	$$
	for $p$-stage algorithms. For adaptive algorithms this constant $\varlimsup_{t\to \infty}\frac{N_{ad}(t, s)}{\log_2t}$ is equal to $s$. In general, our aim is to design a $p$-stage algorithm, which uses asymptotically the same number $s\log_2(t)(1+o(1))$ of tests as the optimal adaptive algorithm.

	\subsection{Related work}
	
	We refer the reader to the monographs~\cite{du2000combinatorial,Cicalese2013Fault-TolerantAlgorithms} for a survey on group testing and its applications. In this paper, only the number of tests needed in the worst-case scenario is considered. For the problem of finding the average number of tests in non-adaptive algorithms we refer the reader to~\cite{freidlina1975design} for $s=O(1)$ and to~\cite{mezard2011group, johnson2019performance} for $s\to\infty$. Also, in paper~\cite{berger2002asymptotic} the average number of tests for 2-stage algorithms was found in model, where each element is defective with probability $p(t)=t^{-\beta+o(1)}$, $\beta\in(0, 1)$.

	Non-adaptive algorithms for the search of at most $s$ defectives can be constructed from $s$-disjunctive (or superimposed) codes~\cite{kautz1964nonrandom, dyachkov1983survey}. Those codes were also investigated under the name of cover-free families~\cite{erdHos1985families}. 
	The best known asymptotic ($s\to\infty$) lower~\cite{d2014lectures} and upper~\cite{d2014bounds} bounds on $C_1(s)$ are as follows
	$$
	\frac{s^2}{4\log_2s}(1+o(1))\leq C_1(s) \leq \frac{s^2}{2\ln2}(1+o(1)).
	$$
	
	Numerical values for small $s$ can be found, for example, in Table 2 in~\cite{d2014lectures}. From these bounds, it follows that for $s\ge 11$ $C_1(s)>s$, i.e. it is impossible to construct a non-adaptive algorithm with asymptotically the same number of tests as in the optimal adaptive algorithm. Also, it is impossible for $s=2$; more precise, in~\cite{coppersmith1998new} the best lower and upper bounds on $C_1(2)$ were established
	$$
	2.0008\leq C_1(2)\leq 3.1898.
	$$
	It is natural to expect that $C_1(s)>s$ for $3\leq s \leq 10$ too, but it hasn't been proved yet.
	
	For the case of $p$-stage algorithms, $p>1$, the only known lower bound is information-theoretic one
	\begin{equation}\label{eq::th_inf_bound}
	C_p(s)\geq s.
	\end{equation}
	
	Group testing algorithms with 2-stages can be obtained from disjunctive list-decoding codes~\cite{dyachkov1983survey} and selectors~\cite{de2005optimal}. Both approaches provide the bound $C_2(s)=O(s)$, but the best results for disjunctive list-decoding codes give a better constant~\cite{d2014bounds}
	\begin{equation}
	C_2(s)\leq (e\ln 2) s(1+o(1)), \quad s\to\infty.
	\end{equation}
	
	In recent work~\cite{v2019twostage} with the help of another approach a new two-stage algorithm was constructed, which outperforms disjunctive list-decoding codes for fixed $s$, but has the same asymptotic for $s\to\infty$. However, for the case of 2 defectives 2 stage algorithm from~\cite{v2019twostage} uses $2\log_2t(1+o(1))$ tests, i.e. $C_2(2)=2$ and the algorithm achieves information-theoretic lower bound on the number of tests.

	This work continues the research started in papers~\cite{d2016hypergraph, v2019twostage}. 
	We prove $C_5(3)=3$ by providing new 5-stage algorithm, which finds 3 defectives using the optimal number ${3\log_2t(1+o(1))}$ of tests.
	
	
	\subsection{Our approach}
	To construct a new algorithm we use a hypergraph framework. Informally, we introduce a $s$-uniform hypergraph $H=(V, E)$, each vertex of which represents one sample. Suppose that we have already carried out some tests. We draw a hyperedge for every $s$-element set of samples, which could be equal to the unknown set of defectives, i.e. it agrees with the outcomes of all tests. Such a hypergraph represents all the information we have obtained from the tests so far. In most of the previous works~\cite{dyachkov1983survey, de2005optimal, v2019twostage}, the first stages of algorithms were constructed in such a way that the hypergraph $H$ would have only a constant amount of hyperedges. It seems that this condition is excessively strong; it requires too many tests at the first stage. In this paper we use such a set of tests for the first stage that the resulting hypergraph is sparse, i.e. the number of hyperedges in $H$ is almost linear on the number of non-isolated vertices. 
	Employing the sparsity of the hypergraph $H$ we explicitly construct subsequent stages to find defectives using approximately $\log_2|E(H)|$ tests. This approach gives us optimal algorithms achieving an information-theoretic lower bound on the number of tests for $s=2, 3$.

	\subsection{Outline}
	
	In Section~\ref{sec::pre}, we introduce the notation and formally describe the hypergraph approach to the group testing problem in general. As a warm-up, in Section~\ref{sec::2Def} we apply new idea to the simplest case $s=2$ to construct a 3-stage algorithm, which uses $2\log_2t(1+o(1))$ tests.
	The main result of the paper is presented in Section~\ref{sec::3Def}, in which the new 5-stage algorithm for finding 3 defectives with an optimal number of tests $3\log_2t(1+o(1))$ is described. Section~\ref{sec::conclusion} concludes the paper.
	
	\section{Preliminaries}\label{sec::pre}
	
	Throughout the paper we use $t$ and $s$ for the number of elements and defectives, respectively. By $[t]$ we denote the set ${\{1, 2\ldots, t\}}$. The binary entropy function $h(x)$ is defined as usual $$h(x)=-x\log_2(x)-(1-x)\log_2(1-x).$$

	A binary $(N \times t)$-matrix with $N$ rows $x_1, \dots, x_N$ and $t$ columns $x(1), \dots, x(t)$ 
	$$
	X = \| x_i(j) \|, \quad x_i(j) = 0, 1, \quad i \in [N],\,j \in [t]
	$$
	is called a {\em binary code of length $N$  and size $t$}.
	The number of $1$'s in the codeword $x(j)$, i.e., $|x(j)| = \sum\limits_{i = 1}^N \, x_i(j)= wN$,
	is called the {\em weight} of $x(j)$, $j \in [t]$ and parameter $w$, $0<w<1$, is the \textit{relative weight}.
	
	We represent $N$ non-adaptive tests with a binary $N\times t$ matrix $X=\|x_{i,j}\|$ in the following way. 
	An entry $x_{i,j}$ equal $1$ if and only if $j$th element is included in $i$th test.
	Let $u \bigvee v$ denote the disjunctive sum of binary columns $u, v \in \{0, 1\}^N$. 
	For any subset $S\subset[t]$ define the binary vector $$r(X,S) = \bigvee\limits_{j\in S}x(j),$$
	which later will be called the \textit{outcome vector}.
	By $S_{un}$, ${|S_{un}|=s}$, denote an unknown set of defectives. 
	
	
	\subsection{Hypergraph framework}\label{subsec::hypApp}
	Let us describe the hypergraph approach to the group testing problem. Suppose that we use a binary $N\times t$ matrix $X$ at the first stage. As a result of performed tests we get the outcome vector ${\y=r(X, S_{un})}$. Construct a hypergraph ${H(X, s, \y)=(V, E)}$ in the following way. The set of vertices $V$ coincides with the set of samples $[t]$. The set of hyperedges consists of all sets $S\subset[t]$, $|S|=s$, such that $r(X, S)=\y$. In other words, the set of hyperedges of the hypergraph $H(X, s, \y)$ represents all possible defective sets of size $s$. We want to design such a matrix $X$ for the first stage of an algorithm that the hypergraph $H(X, s, \y)$ has some good properties, which will allow us to quickly find all defectives at the next few stages.
	
	Previously known algorithms can be described using this terminology. Disjunctive list-decoding codes, selectors and methods from~\cite{v2019twostage} give a binary matrix $X$ such that the hypergraph $H(X, s, \y)$ has only a constant amount of hyperedges for all possible outcome vectors $\y$. Then we can test all non-isolated vertices individually at the second stage. 
	In the algorithm from~\cite{d2016hypergraph} the graph $H(X, 2, \y)$ has a small chromatic number, which also allows finding defectives quickly.

	\section{Algorithm for 2 defectives}\label{sec::2Def}
	
	For the simplest case $s=2$ we propose a 3-stage algorithm with the optimal number of tests $2\log_2t(1+o(1))$.
	\begin{theorem}\label{th::2def}
		$$C_3(2)=2.$$
	\end{theorem}
	\begin{remark}
		It is known~\cite{v2019twostage} that 2 defectives can be found with the optimal number of tests $2\log_2t(1+o(1))$ using only two stages, so, this result is weaker than the result from~\cite{v2019twostage}.
		We present it here only to demonstrate our new approach in the simplest setup. 
	\end{remark}
	\begin{proof}[Proof of Theorem~\ref{th::2def}]
		For the case $s=2$ we deal with a graph instead of hypergraph. 
		Recall that a matching in a graph is a set of non-intersecting edges.
		\begin{definition}
			Call a $N\times t$ matrix $X$ a \textit{2-good} matrix if it satisfies the following properties.
			\begin{enumerate}
				\item For any $\y\in\{0, 1\}^N$ the maximal vertex degree in a graph $H(X, 2, \y)$ is less than $d=\log_2\log_2t$.
				\item For any $\y\in\{0, 1\}^N$, $|\y|=w$, the maximal size of matching in a graph $H(X, 2, \y)$ is less than $10\max(N, t^2q)$, where $q=\frac{\binom{w}{\lfloor pN\rfloor}\binom{\lfloor pN\rfloor}{w-\lfloor pN\rfloor}}{\left(\binom{N}{\lfloor pN\rfloor}\right)^2}$, $p=1-\sqrt{0.5}$.
			\end{enumerate}
		\end{definition}
	Informally, second condition means that for any outcome vector the maximal size of matching can't be a lot bigger than its mathematical expectation, which is equal to $\Theta(t^2 q)$.
		
		\begin{lemma}\label{lem::2def}
			Let $p=1-\sqrt{0.5}$, $d=\log_2\log_2t$,  and
			\begin{equation}
			N=\frac{d+4}{d}\log_2t\frac{1}{h(p)-p}.
			\end{equation}
			Let $X$ be a random $N\times t$ matrix, each column of which is taken independently and uniformly from the set of all columns with $\lfloor pN\rfloor$ ones. Then the probability that the matrix $X$ is 2-good tends to 1 as $t\to \infty$.
		\end{lemma}
		\begin{proof}
			Estimate the probability that for some $\y\in\{0, 1\}^N$, $|\y|=w$, there exists a vertex $v$ with degree at least $d$ in the graph $H(2, X, \y)$. This probability can be upper bounded by the mathematical expectation of the number of sets of $d$ edges $e_i=(v, v_i)$, $i=1,\ldots d$, which is less than
			\begin{multline*}
			\sum\limits_{w=0}^N \binom{N}{w} t^{d+1}\left(\frac{\binom{\lfloor pN\rfloor}{w - \lfloor pN\rfloor}}{\binom{N}{\lfloor pN\rfloor}}\right)^d
			\leq t^{d+3}\max_{w}\left(\frac{\binom{\lfloor pN\rfloor}{w - \lfloor pN\rfloor}}{\binom{N}{\lfloor pN\rfloor}}\right)^d\\<
			t^{d+3}\left(\frac{2^{pN}}{\binom{N}{\lfloor pN\rfloor}}\right)^d=
			t^{d+3}2^{Nd(p-h(p)+o(1))}\\=t^{-1+o(1)}\to 0, \text{ as } t\to\infty.
			\end{multline*}
			In the first inequality we used the fact that $N<2\log_2t$ for $t$ big enough.
			
			In the similar way estimate the probability that for some $\y\in\{0, 1\}^N$, $|\y|=w$, there exists a matching of size $M=10\max(N, t^2q)$ in the graph $H(2, X, \y)$.  Mathematical expectation of the number of such matchings is upper bounded by
			\begin{multline*}
			\sum\limits_{w=0}^N \binom{N}{w} 
			t^{2M}/M!q^M<
			\sum\limits_{w=0}^N \left(\frac{t^2qe}{M}\right)^M\\
			\le 2^N\left(\frac{e}{10}\right)^M\le \left(\frac{e}{5}\right)^N\to 0 \text{ as } t\to\infty.
			\end{multline*}
		\end{proof}
		
		Use 2-good matrix $X$ as a testing matrix at the first stage. Consider an obtained graph $G=(V, E)=H(X, 2, r(X, S_{un}))$. We want to find a partition of all edges $E$ into $M$ disjoint sets $E=\bigsqcup\limits_{i=1}^M E_{i}$ such that
		\begin{enumerate}
			\item There are no intersecting edges in the same set, i.e. if $e_1\in E_{i}$, $e_2\in E_{j}$ and $e_1\cap e_2\ne \emptyset$, then $i\ne j$.
			\item There are no edges $e_1$ and $e_2$ in the same set, such that there exists an edge $e\in E$, which intersects both $e_1$ and $e_2$, i.e. if $e_1\in E_{i}$, $e_2\in E_{j}$, $e\in E$, $e_1\cap e\ne \emptyset$, $e_2\cap e\ne \emptyset$, then $i\ne j$.
		\end{enumerate}
		Every edge $e$ can't be in the same set with its adjacent edges and also the edges, which have a common adjacent edge with $e$. Since the degree of every vertex is less than $d$, we conclude that the number of such edges is less than $2d^2$. Hence, we can construct such partition greedily for $M=2d^2$.
		
		At the second stage, we carry out $2M$ tests. For each $i\in[M]$ two sets of vertices $S_{2i-1}$ and $S_{2i}$ are tested. The set $S_{2i-1}$ consists of all vertices incident to edges from $E_i$, set $S_{2i}$ is equal to $V\setminus S_{2i-1}$. We claim that the responses to tests $2i-1$ and $2i$ are equal to 1 and 0 respectively if and only if the set of defectives $S_{un}$ coincides with an edge from $E_i$. Indeed, if $S_{un}=e\in E_i$ then outcomes are equals to 1 and 0. Otherwise, $S_{un}$ can intersect at most one edge from $E_i$, therefore, the result of test $2i$ is positive.
		
		So, after the second stage, we will find a set $E_i$, which contains the defective edge. We can treat each edge from this set as a separate sample, only one of which is defective. Therefore, the defective edge can be found at the third stage by binary search using at most $\lceil \log_2|E_i|\rceil\le \log_2|E|+1$ tests. This step finishes the algorithm.
		
		The total number of tests is upper bounded by $N+2M+\log_2|E|+1=N+\log_2|E|+o(\log_2t)$. Let us estimate the cardinality of $E$.
		
		Consider some maximal matching in the graph $G$. Every edge is incident to at least one vertex from this matching, the degree of each vertex is less than $d$, therefore, 
		$$
		|E|\leq 2d\cdot 10\max(N, t^2q).
		$$
		
		Since $N<2\log_2t$ we conclude that $N+\log_2(20dN)<2\log_2t(1+o(1))$. Hence, it is sufficient to show that 
		$$
		N+\log_2(t^2q)\le 2\log_2t(1+o(1)).
		$$
		
		Indeed, 
		\begin{multline*}
		N+\log_2(t^2q)=2\log_2t+N+\log_2q
		\\\leq 2\log_2t+N+N\max_{\omega}(\omega h(p/\omega) + ph((\omega-p)/p)-2h(p)+o(1))
		\end{multline*}
		The expression $\omega h(p/\omega) + ph((\omega-p)/p)-2h(p)$ attains its maximum -1 at $\omega=0.5$. Therefore, the total number of tests is at most
		$$
		2\log_2t+o(\log_2t).
		$$
		Theorem~\ref{th::2def} is proved.

	\end{proof}

	\section{Algorithm for 3 defectives}\label{sec::3Def}
	\begin{theorem}\label{th::3def}
		$$C_5(3)=3.$$
	\end{theorem}
	We prove Theorem~\ref{th::3def} by presenting a new algorithm for finding 3 defectives, which uses $3\log_2t(1+o(1))$ tests. It is a first multistage algorithm for $s=3$ with the optimal number of tests. The best previously known algorithm~\cite{v2019twostage} used approximately $3.10\log_2t(1+o(1))$ tests. Proofs of Lemmas from this Section are postponed to Appendix.
	\begin{proof}[Proof of Theorem~\ref{th::3def}]
		
		To construct a matrix for the first stage of our algorithm  we must introduce some useful terminology. 
		Fix an integer $L$ and consider a $s$-uniform hypergraph $H$. Call the set of hyperedges $e_1, e_2, \ldots, e_L$ a $(s, k)$ configuration of size $L$ 
		if $e_i\cap e_j=U$, $|U|=k$, for any $i$ and $j$. In other words, $(s, k)$ configuration consists of $L$ hyperedges such that the intersection of every two hyperedges is the same set of size $k$. 
		

		We construct a matrix for the first stage of our algorithm randomly. More precisely, we take a binary matrix $X$ of size $N\times t$, in which each column is taken independently and uniformly from the set of all columns with $\lfloor pN\rfloor$ ones. Let $\Pr_1(s, w)$ be equal to the probability that the union of $s$ columns from such ensemble equals to a fixed vector with $w$ ones. 
		Let $\Pr_2(s, w_1, w)$ be equal to the probability that the union of $s$ columns with a fixed column $\y_1$ of weight $w_1$ equals to a fixed vector $\y$ of weight $w$, $\y\bigvee \y_1=\y$. 
		
		\begin{definition}\label{def::3good}
			Call a $N\times t$ matrix $X$ a 3-good matrix if it satisfies the following list of properties.
			
			\begin{enumerate}
				\item The hypergraph $H(X, 3, \y)$ doesn't have $(3, 1)$ configurations of size $L_1=\log_2\log_2 t$ for any vector $\y\in\{0,1\}^N$.
				\item The hypergraph $H(X, 3, \y)$ doesn't have $(3, 0)$ configurations of size $10\max(t^3\Pr_1(3, w), N)$ for any vector $\y\in\{0,1\}^N$, $w=|\y|$.
				\item Let $\y$ and $\y_1$ be two binary vectors of length $N$, $\y\bigvee \y_1=\y$, $|\y_1|=w_1$, $|\y|=w$. For any such vectors $\y$ and $\y_1$ the number of columns $\z$ in $X$ such that $\y_1\bigvee \z=\y$ is less than $10 B(N, t)$, where $B(N, t)$ is defined as follows
				$$
				B(N, t)=
				\begin{cases}
				t\Pr_2(1, w_1, w), \text{ if }  t\Pr_2(1, w_1, w)>N;\\
				N, \text{ if }  t^{-\frac{1}{\sqrt{L_1}}}\le t\Pr_2(1, w_1, w)\le N;\\
				L_1/10, \text{ if }  t\Pr_2(1, w_1, w)< t^{-\frac{1}{\sqrt{L_1}}}.\\
				\end{cases}$$
				\item Let $\y$ be a binary vector of length $N$, $|\y|=w$, $w_1$ is some integer, $w_1\le w$. Then the number of non-intersecting pairs of columns $\z_1$, $\z_2$ from matrix $X$ such that $\z_1\bigvee \z_2\bigvee \y=\y$,  $|\z_1\bigvee \z_2|=w_1$, is less than $10\max\left(N, \binom{w}{w_1}t^2\Pr_1(2, w_1)\right)$.
			\end{enumerate}
		\end{definition}
	
		The first property of $3$-good matrix ensures that the hypergraph $H(X, 3, \y)$ would be sparse. Other properties guarantees that sizes of some interesting configurations of edges are not very different from their mathematical expectations.
		
		Define $A_1(s, \omega)$ and $A_2(s, \omega_1, \omega)$ as follows
		\begin{equation}
		A_1(s, \omega)=\lim\limits_{N\to\infty}\frac{-\log_2 \Pr_1(s, \lfloor \omega N \rfloor)}{N};
		\end{equation}
		\begin{equation}
		A_2(s, \omega_1, \omega)=\lim\limits_{N\to\infty}\frac{-\log_2 \Pr_2(s, \lfloor\omega_1 N\rfloor, \lfloor \omega N \rfloor)}{N}
		\end{equation}

		\begin{lemma}\label{lem::3Def1st}
			Let $p=1-0.5^{1/3}$, $L_1=\log_2\log_2t$ and
			\begin{equation}
			N=\frac{2L_1+10}{L_1}\log_2t\max_{p\le\omega\le 3p}    \frac{1}{A_2(2, p, \omega)}.
			\end{equation}
			Let $X$ be a random $N\times t$ matrix, each column of which is taken independently and uniformly from the set of all columns with $\lfloor pN\rfloor$ ones. Then the probability that the matrix $X$ is 3-good tends to 1 as $t\to \infty$.
		\end{lemma}

		
		We use a 3-good matrix $X$ as a testing matrix at the first stage of our algorithm. Consider an obtained hypergraph $H=H(X, 3, r(X, S_{un}))=([t], E)$. Introduce a new graph $G'=([t], E')$. The set of vertices coincides with the set of samples. Two vertices $v_1$ and $v_2$ are connected with an edge if there exists at least $L_2=3L_1$ hyperedges $e\in E$ from the hypergraph $H$, such that $v_1\in e$ and $v_2\in e$.
		
		\begin{lemma}\label{lem::3defVerDeg}
			The degree of every vertex in the graph $G'$ is less than $L_1$.
		\end{lemma}
		
		Divide all hyperedges $E$ of the hypergraph $H$ into two groups $E_1$ and $E_2$, $E=E_1\bigsqcup E_2$. We put a hyperedge into $E_1$ if it contains an edge from $G'$ as a subset; otherwise, we put a hyperedge into $E_2$. Note that the hypergraph $H_2=([t], E_2)$ can't contain a $(3, 2)$ configuration of size $L_2$. The following lemma shows that the hypergraph $H_2$ is quite sparse.
		\begin{lemma}\label{lem::3defH2}
			The following two claims hold.
			\begin{enumerate}
				\item The degree of each vertex in $H_2=([t], E_2)$, i.e. the number of hyperedges containing one vertex, is at most $2L_1L_2$.
				\item The number of hyperedges in $E_2$ is less than the size of the biggest $(3, 0)$ configuration multiplied by $6L_1L_2$.
			\end{enumerate}
			
		\end{lemma}
		
		
		For every hyperedge $e=(v_1, v_2, v_3)\in E_1$ we choose one vertex $v_i$ such that $e\setminus v_i\in E'$. Call that vertex an additional vertex of the hyperedge $e$. If there are multiple ways to choose such vertex, we do it arbitrarily.
		
		Introduce a new directed graph $G''=([t], E'')$. For every hyperedge $e\in E_1$, $e=(v_1, v_2, v_3)$ with an additional vertex $v_1$ we add to $E''$ 4 arcs $(v_1, v_2)$, $(v_1, v_3)$, $(v_2, v_3)$, $(v_3, v_2)$. If an arc has already been in $E''$, we don't add it second time, i.e. there is no multi-edges in $G''$.
		\begin{lemma}\label{lem::3defDirGr}
			The out-degree in the graph $G''$ is less than $3L_1^2$. 
		\end{lemma}
		At the second stage we want to check whether the set of defectives lies in $E_1$ or $E_2$. 
		
		\begin{lemma}\label{lem::3defPart}
			There exists a partition of $E_2$ into $M\le 96L_1^2L_2^2$ disjoint sets $E_2=\bigsqcup\limits_{i=1}^M E_{2,i}$ such that
			\begin{enumerate}
				\item There are no intersecting hyperedges in the same set, i.e. if $e_1\in E_{2,i}$, $e_2\in E_{2,j}$ and $e_1\cap e_2\ne \emptyset$, then $i\ne j$.
				\item There are no hyperedges $e_1$ and $e_2$ in the same set, such that there exists a hyperedge $e\in E$, which intersects both $e_1$ and $e_2$, i.e. if $e_1\in E_{2,i}$, $e_2\in E_{2,j}$, $e\in E$, $e_1\cap e\ne \emptyset$, $e_2\cap e\ne \emptyset$, then $i\ne j$.
			\end{enumerate}
		\end{lemma}
		\begin{remark}
			We emphasize that in the second condition a hyperedge $e$ is not necessarily from $E_2$, it can be from $E_1$ as well.
		\end{remark}

		Construct a partition from Lemma~\ref{lem::3defPart}. The second stage consists of $2M$ tests. For each $i\in [M]$ two sets of vertices $S_{2i-1}$ and $S_{2i}$ are tested. In the first tested set $S_{2i-1}$ we include all vertices $v$, which belong to some hyperedge $e\in E_{2,i}$. In the second test all other vertices are included, $S_{2i}=[t]\setminus S_{2i-1}$.
		
		If the unknown set of defectives coincides with some hyperedge $e\in E_{2,i}$, then the outcomes of tests $2i-1$ and $2i$ are $1$ and $0$ respectively. Otherwise, the outcomes are different. The first claim is obvious. To prove the second claim note that $S_{un}$ can't intersect two hyperedges from $E_{2,i}$ by Lemma~\ref{lem::3defPart}; therefore, it can't be a subset of $S_{2i-1}$, which means that the outcomes can't be equal to 1 and 0 respectively.
		
		So, we have 2 cases.
		\begin{enumerate}
			\item There is an integer $i$ such that the outcomes for tests $S_{2i-1}$ and $S_{2i}$ are equal to $1$ and $0$ respectively.
			
			In that case $S_{un}=e\in E_{2,i}$. Then we can think about each hyperedge from $E_{2,i}$ as a separate sample. This set of samples contains exactly one defective element $e$, which can be found by using a binary search algorithm.
			
			To sum up, in this case we have used 3 stages and $N+2M+\lceil\log_2|E_{2,i}|\rceil\le N+o(\log t)+\log_2|E_2|$ tests.
			
			\item There is no integer $i$ such that the outcomes for tests $S_{2i-1}$ and $S_{2i}$ are equal to $1$ and $0$ respectively.
			
			It means that $S_{un}$ coincides with some hyperedge in $E_1$. Recall the graph $G''$. Let $V_0$ be a set of all isolated vertices in $G''$. By Lemma~\ref{lem::3defDirGr} the out-degree of every vertex in this graph is less than $3L_1^2$, therefore, it is possible to partition the set of all non-isolated vertices into $q\leq 6L_1^2$ disjoint sets $V_i$, $V\setminus V_0=\bigsqcup\limits_{i=1}^q V_i$ such that there is no arc $e\in E''$ inside one set $V_i$. There is an arc in at least one direction between any two vertices from the edge $S_{un}=e=(v_1, v_2, v_3)\in E_1$, hence, 3 vertices $v_1$, $v_2$, $v_3$ will be placed in 3 different sets.
			
			At the third stage, we test each set $V_i$ separately. We will obtain exactly 3 positive outcomes at this stage. Without loss of generality assume that the tested set $V_1$ has given a positive result.
			This set contains exactly 1 defective element. At the fourth stage find this vertex using $\lceil\log_2|V_1|\rceil\le \lceil\log_2|V\setminus V_0|\rceil\le \log_2|E_1|+3$ tests by binary search. Denote this vertex as $v$.
			
			Vertex $v$ is an additional vertex for hyperedges $e_1, \ldots, e_{M_1}\in E_1$. By Lemma~\ref{lem::3defDirGr} $M_1<3L_1^2$. Also, vertex $v$ belongs to hyperedges $e'_1, \ldots, e'_{M_2}$, $e'_i=(v, v'_i, u_i)$, $u_i$ is an additional vertex for the edge $e'_i$. Define sets of vertices $W=\cup_{i=1}^{M_1}e_i\setminus \{v\}$, $V'=\cup_{i=1}^{M_2}{v'_i}$, $U=\cup_{i=1}^{M_2}{v'_i}\setminus V'$. At the fifth stage we perform $|W|+|V'|+\lceil\log_2|U|\rceil$ tests. Each element of $W$ and $V'$ is tested separately; binary search is performed on $U$ to find one defective element. If the vertex $v$ is additional in the hyperedge $S_{un}$, then two others defectives will be found in $W$. Otherwise, at least one defective elements would be found in $V'$. If there is exactly one defective in $V'$, the last one will be found in $U$. This stage completes the algorithm.
			
			To sum up, in this case we have used 5 stages and at most $N+2M+q+\log_2|E_1|+3+|W|+|V'|+\log_2|U|+1$. Recall that $M<96L_1^2L_2^2$, $q<6L_1^2$; $|W|<6L_1^2$ by Lemma~\ref{lem::3defDirGr}, $|V'|<L_1$ by Lemma~\ref{lem::3defVerDeg}, $|U|\le 3|E_1|$; therefore, the total number of tests is upper bounded by $N+2\log_2|E_1|+o(\log t)$.

		\end{enumerate}
		
		The following Lemma finishes the proof of Theorem~\ref{th::3def}.
		\begin{lemma}\label{lem::3defNumEdges}
			\begin{equation}\label{eq::E2}
			N+\log_2|E_2|=3\log_2t(1+o(1));
			\end{equation}
			\begin{equation}\label{eq::E1}
			N+2\log_2|E_1|<3\log_2t(1+o(1)).
			\end{equation}
			
		\end{lemma}
	\end{proof}

	\section{Conclusion}\label{sec::conclusion}
	
	A new approach to construct multistage group testing procedures was considered. It allows to design 3-stage and 5-stage algorithms with optimal values of $C_3(2)=2$ and $C_5(3)=3$  for the cases $s=2$ and $s=3$ respectively. The algorithm with the optimal number of tests for $s=3$ was obtained for the first time.
	
	The natural open problem is to generalize this approach to the case $s>3$ to construct algorithms with $C_p(s)=s$. Another possible direction is to prove a lower bound on $C_p(s)$ for some $p>1$, which is stronger than information-theoretic bound $s$.

	
	\section{Acknowledgement}
	
	The reported study was funded by RFBR and JSPS, project number 20-51-50007, and by RFBR through grant no. 18-31-00361~MOL\_A.
	
	
	\bibliographystyle{IEEEtran}
	\bibliography{Multistage}
	\appendix
	\section{Proofs of Lemmas}
\begin{proof}[Proof of Lemma~\ref{lem::3Def1st}]
	Denote an event that a property $i$ from Definition~\ref{def::3good}, $1\leq i\leq 5$, is violated as $B_i$. Let us estimate probabilities of these events from above.
	\begin{enumerate}
		\item
		Fix a vector $\y\in\{0,1\}^N$, $|\y|=w$, $\omega=w/N$. Denote an event that we have a $(3, 1)$ configuration of size $L_1$ in $H(X, 3, \y)$ as $B_{1,\y}$.
		Let $Y_{\y}$ be a random variable equals to the number of $(3,1)$ configurations of size $L_1$. We upper bound the probability of $B_{1,\y}$ by the mathematical expectation of $Y_{\y}$, i.e. $\Pr(B_1)\leq \sum\limits_{\y\in\{0,1\}^N}\Pr(B_{1,\y})\leq \sum\limits_{\y\in\{0,1\}^N}\mathbb{E} Y_{\y}$. To estimate $\mathbb{E} Y_{\y}$ we represent it as a sum of $\binom{t}{2L_1+1}\frac{(2L_1+1)!}{L_1!2^{L_1}}<t^{2L_1+1}$ indicators corresponding to all possible $(3,1)$ configurations of size $L_1$.
		
		\begin{align*}
		\mathbb{E} Y_{\y}<t^{2L_1+1}(\Pr\nolimits_2(2, p, w))^{L_1}
		\\=t^{2L_1+1}2^{-NL_1(A_2(2, p, \omega)+o(1))}
		\\\leq
		t^{2L_1+1}2^{-\log_2t(2L_1+10)(1+o(1))}\\<t^{-9+o(1)}.
		\end{align*}
		For $p=1-0.5^{1/3}$, $\max_{p\le\omega\le 3p}    \frac{1}{A_2(2, p, \omega)}\approx 1.35$, therefore, $N<(3+o(1))\log_2t$;
		\begin{align*}
		\Pr(B_1)\leq 2^N\max_{\y}\mathbb{E} Y_{\y}<t^{-6+o(1)}\to 0.
		\end{align*}
		\item
		Let $Y_{\y}$ be a random variable equals to the number of $(3,0)$ configurations of size $M=10\max(t^3\Pr_1(3, w), N)$; then $\Pr(B_2)\leq 2^N\max\limits_{\y}\mathbb{E} Y_{\y}$.
		\begin{align*}
		\mathbb{E} Y_{\y}<\binom{t}{3M}\frac{(3M)!}{M!(3!)^M}(\Pr\nolimits_1(3, w))^M
		\\\leq
		\frac{t^{3M}}{M!6^M}(\Pr\nolimits_1(3, w))^M<\left(\frac{t^3\Pr_1(3, w)e}{6M}\right)^M\\
		\leq \left(\frac{e}{60}\right)^M<\left(\frac{e}{60}\right)^N.
		\end{align*}
		\item
		Fix two binary vectors $\y$ and $\y_1$ of length $N$ such that $\y\bigvee\y_1=\y$, $|\y_1|=w_1$, $|\y|=w$. Let $Y_{\y, \y_1}$ be a random variable equals to the number of columns $z$ in $X$ such that $\y_1\bigvee z=\y$. 
		\begin{align*}
		\Pr(B_4)\leq 4^N\max\limits_{\y_1,\y}\Pr(Y_{\y, \y_1}>10B(N, t))\\
		\leq 4^N\max\limits_{\y_1,\y}\frac{t^{10B(N, t)}}{(10B(N, t))!}(\Pr\nolimits_2(1, w_1, w))^{10B(N, t)}\\
		\leq
		4^N\max\limits_{\y_1,\y}\left(\frac{t\Pr\nolimits_2(1, w_1, w)e}{10B(N, t)}\right)^{10B(N, t)}.
		\end{align*}
		\begin{enumerate}
			
			\item 
			If $t\Pr_2(1, w_1, w)\geq t^{-\frac{1}{\sqrt{L_1}}}$, then $\frac{t\Pr\nolimits_2(1, w_1, w)}{B(N, t)}\leq 1$ and $B(N, t) \geq N$, hence
			\begin{align*}
			4^N\max\limits_{\y_1,\y}\left(\frac{t\Pr\nolimits_2(1, w_1, w)e}{10B(N, t)}\right)^{10B(N, t)}\\\leq
			4^N\left(\frac{e}{10}\right)^{10N}\to 0.
			\end{align*}
			\item 
			If $t\Pr_2(1, w_1, w)< t^{-\frac{1}{\sqrt{L_1}}}$, then $\frac{t\Pr\nolimits_2(1, w_1, w)}{B(N, t)}\leq 1$ and $B(N, t) \geq N$, hence
			\begin{align*}
			4^N\max\limits_{\y_1,\y}\left(\frac{t\Pr\nolimits_2(1, w_1, w)e}{10B(N, t)}\right)^{10B(N, t)}\\\leq
			4^N\max\limits_{\y_1,\y}\left(t\Pr\nolimits_2(1, w_1, w)\right)^{10B(N, t)}\\\leq
			4^Nt^{-\frac{L_1}{\sqrt{L_1}}}<t^{6-\sqrt{L_1}}\to 0.
			\end{align*}
		\end{enumerate}
			
		\item
		Let $Y_{\y, w_1}$ be a random variable equals to the number of sets $S$ of cardinality $M=10\max\left(N, \binom{wN}{w_1N}\Pr_1(2, w_1)\right)$, consisting of pairs of non-intersecting columns $\z_1$, $\z_2$ from matrix $X$ such that $\z_1\bigvee \z_2\bigvee \y=\y$,  $|\z_1\bigvee \z_2|=w_1$; then $\Pr(B_5)\leq N2^N\max\limits_{\y, w_1}\mathbb{E} Y_{\y, w_1}$.
		
		\begin{align*}
		&N2^N\max\limits_{\y, w_1}\mathbb{E} Y_{\y, w_1}\\
		&<N2^N\max\limits_{w, w_1}t^{2M}/M!\left(\binom{w}{w_1}\Pr_1(2, w_1)\right)^M\\
		&\leq N2^N\max\limits_{w, w_1}\left(\frac{t^2\binom{w}{w_1}\Pr_1(2, w_1)e}{M}\right)^M\\
		&\leq
		N2^N\left(\frac{e}{10}\right)^N\to 0.
		\end{align*}
		
	\end{enumerate}
\end{proof}

\begin{proof}[Proof of Lemma~\ref{lem::3defVerDeg}]
	Seeking for a contradiction assume that vertex $v$ has degree at least $L_1$ in graph $G'$. It means that there exist $L_1$ vertices $v_1, \ldots, v_{L_1}$, and $L_1L_2$ hyperedges $e_{i,j}=(v, v_i, v_{i, j})\in E$, $1\leq i\leq L_1$, $1\leq j\leq L_2$. We show that in this case there is a $(3, 1)$ configuration of size $L_1$; more precise, it is possible to find a set of $L_1$ hyperedges $e_{1, j_1}, e_{2, j_2}, \ldots, e_{L_1, j_{L_1}}$ such that any two of these hyperedges have only one common vertex $v$. Indeed, we can construct such set by choosing hyperedges one by one from $i=1$ to $i=L_1$. At each step $i$ we have $L_2$ candidates, with at most $L_1+ (i-1)<L_2$ of which are prohibited; the existence of such $(3, 1)$ configuration contradicts the first property from definition~\ref{def::3good}.
\end{proof}

\begin{proof}[Proof of Lemma~\ref{lem::3defH2}]
	 Let $v$ be a vertex of $H_2$ with a degree at least $M=2L_1L_2$. Consider hyperedges $e_1=(v, v_1, u_1), \ldots, e_{2L_1L_2}=(v, v_{M}, u_{M})$. Construct a maximal $(3,1)$ configuration $e_{i1}, \ldots, e_{M_1}$, consisting of these hyperedges. Its size $M_1$ is less than $L_1$ by the property~1 from Definition~\ref{def::3good}. Consider pairs of vertices $(v, v_{i_1}), (v, u_{i_1}), \ldots, (v, v_{i_{M_1}}), (v, u_{i_{M_1}})$. Every hyperedge $e_i$ contains at least one such pair as a subset due to the fact that the constructed $(3, 1)$ configuration is maximal. From the other hand, no pair can be included in $L_2$ hyperedges. Indeed, in this case we would have a $(3, 2)$ configuration of size $L_2$, which contradicts the definition of the hypergraph $H_2$. Therefore, $M< L_2\cdot 2L_1$, and the first claim is proved. 
	
	The second claim is an immediate consequence of the first one. Indeed, consider the biggest $(3, 0)$ configuration in the hypergraph $H_2$. Say that it has the cardinality $M'$. Every hyperedge of $H_2$ has at least one common vertex with such configuration, therefore, the total number of hyperedges in $H_2$ is at most $3M'\cdot 2L_1L_2=M'\cdot 6L_1L_2$.
\end{proof}
\begin{proof}[Proof of Lemma~\ref{lem::3defDirGr}]
	Fix a vertex $v$. Denote as $S$ the set of edges $e\in E_1$, such that $v$ is an additional vertex for an edge $e$. Then $\deg_{out}(v)< L_1+2|S|$. Indeed, if $v_1$ is an additional vertex for an edge $e=(v, v_1, v_2)$, then we add only one arc $(v, v_2)$ from $v$. Moreover, $(v, v_2)$ is an edge of the graph $G'$. By Lemma~\ref{lem::3defVerDeg} the number of such arcs is less then $L_1$.
	
	Estimate the cardinality of the set $S=\{e_1, e_2,\ldots, e_{|S|}\}$. The reasoning is very similar to the proof of Lemma~\ref{lem::3defH2}. Construct a maximal $(3,1)$ configuration $e_{i1}, \ldots, e_{M_1}$, consisting of edges from $S$. The size of this configuration is less than $L_1$ by the property~1 from Definition~\ref{def::3good}. 
	Consider pairs of vertices $(v, v_{i_1}), (v, u_{i_1}), \ldots, (v, v_{i_{M_1}}), (v, u_{i_{M_1}})$. Every edge $e_i$ contains at least one such pair as a subset by construction. From Lemma~\ref{lem::3defVerDeg} we conclude, that no pair can be included in $L_1$ edges. Hence, $|S|<2L_1^2$.
	
	So, the total out-degree is less than $2L_1^2+L_1\leq 3L_1^2$.
	
\end{proof}
\begin{proof}[Proof of Lemma~\ref{lem::3defPart}]
	Introduce a new directed graph $G'''(E_2, E''')$. The vertex set of this graph coincides with the set $E_2$. We draw an arc $(e_1, e_2)$, $e_1, e_2\in E_2$ in three cases
	\begin{enumerate}
		\item $e_1\cap e_2\ne \emptyset.$
		\item There is a hyperedge (or edge) $e\in E_2\cup E'$, such that $e_1\cap e\ne \emptyset,$ $e_2\cap e\ne \emptyset.$
		\item There is a hyperedge $e=(v_1, v_2, v_3)\in E_1$ with an additional vertex $v_1$, such that
		$v_1\in e_1$, $e\cap e_2\ne \emptyset$.
	\end{enumerate}
	Fix a hyperedge $e\in E_2$. Estimate its out-degree. By Lemma~\ref{lem::3defH2}, the first condition gives at most $6L_1L_2$ arcs, and the second -- at most $24L_1^2L_2^2+6L_1^2L_2$. From Lemmas~\ref{lem::3defH2},~\ref{lem::3defDirGr}, we conclude that the third condition gives at most $18L_1^2\cdot 2L_1L_2<18L_1^2L_2^2$ arcs. The total out-degree is less than $48L_1^2L_2^2$, therefore, it is possible to partition vertices of $G'''$ into $96L_1^2L_2^2$ sets in such a way that the endpoints of every arc are in different sets. It is readily seen that this partition of hyperedges from $E_2$ has all required by Lemma~\ref{lem::3defPart} properties.
\end{proof}
\begin{proof}[Proof of Lemma~\ref{lem::3defNumEdges}]
	To prove the first inequality we estimate the cardinality of $E_2$. From the second property of Definition~\ref{def::3good} and Lemma~\ref{lem::3defH2} we conclude that
	$$
	|E_2|\leq \max_{w}60L_1L_2\max(t^3\Pr\nolimits_1(3, w), N)
	$$
	and
	$$
	\log_2|E_2|\leq \max(0, \max_{w}\log_2\Pr\nolimits_1(3, w)+3\log_2t) + o(\log_2t).
	$$
	Since $N<3\log_2t$ it is sufficient to verify that $N+\max\limits_{w}\log_2\Pr_1(3, w)<o(\log_2t)$.
	
	To estimate $\Pr_1(3, w)$ consider a different ensemble of random matrices, in which each entry is chosen independently and equals 1 with probability $p$. Let $\Pr_1'(3, w)$ be equal to the probability that union of three columns equals a fixed vector of weight $w$. It is readily seen that 
	$\Pr_1'(3, w)=(0.5+o(1))^N$. From the other hand, 
	$$
	\Pr\nolimits_1'(3, w)>\left(\binom{N}{\lfloor pN\rfloor}p^{\lfloor pN\rfloor}(1-p)^{n-\lfloor pN\rfloor}\right)^3\Pr\nolimits_1(3, w),
	$$
	therefore,
	$$
	\log_2\Pr\nolimits_1(3, w)<\log_2\Pr\nolimits_1'(3, w)+o(N)=-N+o(N).
	$$
	

	So,
	$$N+\max\limits_{w}\log_2\Pr\nolimits_1(3, w)<o(N)=o(\log_2t)$$
	which proves the inequality~\eqref{eq::E2}.
	
	Estimate the cardinality of $E_1$ using properties~3, 4 of Definition~\ref{def::3good} as follows
	
	\begin{align*}
	|E_1|\leq \max_{w_1, w}10\max\left(N, \binom{w}{w_1}t^2\Pr\nolimits_1(2, w_1)\right)\cdot 10B(N, t);
	\end{align*}
	
	\begin{align*}
	\log_2|E_1|\leq \max_{\omega_1, \omega}\left[\max(0, r_1) + \max(0, r_2)\right]+ o(\log_2t),
	\end{align*}
	where $$r_1=N(\omega h(\omega_1/\omega) +\omega_1h(p/\omega_1)+ph((\omega_1-p)/p) - 2h(p)) + 2\log_2t,$$
	$$r_2=\log_2t+N(\omega_1h((\omega-p)/\omega_1)-h(p)).$$ 
	
	If $t\Pr_2(1, w_1, w)<t^{\frac{-1}{\sqrt{L_1}}}$ then by property~3 of Definition~\ref{def::3good} we have less than $L_1$ additional columns for every column of weight $w_1$; therefore, the set of hyperedges $E_1$ is empty for such $w_1$ and $w$.
	
	From this moment consider only $w_1$ and $w$ such that $t\Pr_2(1, w_1, w)\geq t^{\frac{-1}{\sqrt{L_1}}}$. For such parameters $r_2\geq o(\log_2t)$, therefore,
	\begin{multline*}
	N+2\log_2|E_1|\\
	\leq N +2\max_{\{\omega_1, \omega|r_2\geq 0\}}\left[\max(r_2, r_1+r_2)\right]+ o(\log_2t).
	\end{multline*}
	We find that maximum numerically and obtain
	\begin{multline*}
	N+2\log_2|E_1|
	<2.965\log_2t+o(\log_2t)<3\log_2t(1+o(1)).
	\end{multline*}

\end{proof}

\end{document}